\documentclass[%
 reprint,
 aps,
prl,
]{revtex4-2}
\usepackage{dcolumn}
\usepackage{graphicx,xcolor} 

\usepackage{amsmath,amssymb,amsfonts,mathtools,amsthm}
\usepackage{hyperref}

\usepackage[capitalize]{cleveref}

\usepackage{physics,bm}

\newtheorem{theorem}{Theorem}

\newtheorem{lemma}[theorem]{Lemma}

\newcommand{\co}{\mathcal{O}}

\newcommand{\dt}{{\Delta t}}

\begin{document}

\setlength{\abovedisplayskip}{2pt}
\setlength{\belowdisplayskip}{2pt}
\setlength{\abovedisplayshortskip}{2pt}
\setlength{\belowdisplayshortskip}{2pt}

\title{
Auxiliary-Field Quantum Monte Carlo on Quantum Hardware via Unitary Dilation}
\author{Xiantao Li} \affiliation{The Pennsylvania State University}
\thanks{Xiantao.Li@psu.edu}

\begin{abstract}
We present near-term quantum algorithms for auxiliary-field quantum Monte Carlo (AFQMC), viewed as imaginary-time
projection for ground-state calculation as an ensemble of one-body propagators driven by stochastic fields $\Omega$. Starting from the
Feynman-Kac formula, we convert each trajectory into a sequence of piecewise-constant one-body generators using stochastic Magnus expansions up to second order, and embed the resulting nonunitary slices into unitaries with a
small ancilla overhead. This lifts the projector dynamics to a unitary evolution,
enabling coherent circuit execution in the regime $\|\Omega \| \tau=\co(1)$ and reducing the need for frequent
mid-circuit measurement. We further derive an equivalent linear-combination-of-unitaries (LCU) form that
yields system-only, shallower circuits by trading ancilla cost for additional trajectory sampling. Benchmarks on the
Hubbard model verify the accuracy of the dilation and Magnus expansions classically and demonstrate multi-step executions
on IBM quantum hardware.
\end{abstract}
\maketitle

\section{Introduction.}
Strongly correlated electronic structure is one of the most compelling applications of quantum computation. In the fault-tolerant setting, phase-estimation and qubitization based approaches, often accelerated by low-rank factorizations of the two-electron integrals such as tensor hypercontraction, offer favorable asymptotic scaling  \cite{Babbush2018_PRX_Qubitization,Lee2021PRXQTHC,BauerBravyiMottaChan2020_ChemRev}. However, compiled resource estimates for chemically relevant problems still require long logical runtimes and large physical-qubit overheads  beyond current devices. In parallel, near-term variational strategies can reduce coherent circuit depth, but their accuracy depends on ansatz expressibility and the ability to optimize non-convex loss \cite{Peruzzo2014_VQE,Kandala2017_IBM_HEVQE,McClean2016_VQETheory}.

At the fault-tolerant level, ground-state preparation can also be achieved by spectral filtering and related imaginary-time inspired constructions, with near-optimal dependence on the spectral gap and target accuracy \cite{LinTong2020nearoptimalground,DongLinTong2022QETU}. 
Another important route is imaginary-time evolution (ITE), which avoids variational optimization by filtering excited-state components. Motta \emph{et al.} introduced quantum ITE schemes and quantum Lanczos as an ansatz-free framework for ground- and thermal-state preparation \cite{Motta2020NatPhys}, while variational and measurement-based ITE schemes broaden the algorithmic landscape \cite{McArdle2019,Mao2023MeasurementBasedDeterministicITE,Sun2021}. These methods, however, typically rely on measurement-intensive local reconstructions or variational parameter updates.

Auxiliary-field quantum Monte Carlo (AFQMC) is one of the most accurate and scalable classical methods for interacting fermionic systems, especially in strongly correlated materials and molecules. Through the Hubbard-Stratonovich (HS) transformation, it rewrites quartic interactions as a stochastic average over one-body propagators driven by auxiliary fields \cite{Hubbard1959,Stratonovich1958,Hirsch1983_AFQMC,ZhangCarlsonGubernatis1997_CPMC,MottaZhangChan2019_AFQMCsym,ShiZhang2021_AFQMC_Review}. Because each sampled propagator maps Slater determinants (SDs) to SDs, this structure is highly attractive for quantum algorithms. Yet this structural advantage has not previously led to an explicit gate-based realization of AFQMC algorithms. The obstacle is that each auxiliary-field slice \(\Omega\) is generally nonunitary and trajectory dependent, while the ITE emerges only after averaging over the trajectory ensemble. A naive implementation would therefore require either measurement-heavy adaptive updates or a coherent summation over trajectories with substantial postselection and depth overhead.

This letter closes this gap by formulating the first gate-level realization of HS-driven AFQMC projector dynamics on quantum hardware, by developing a unitary dilation of nonunitary one-body propagators. The resulting  dynamics can be executed coherently while preserving SD-to-SD structure, which can be compiled using standard Givens-rotation decompositions  \cite{Jiang2018_FermionicGaussian,Kivlichan2018_LinearDepthSlater}. We construct a resource-efficient dilation primitive in the regime \(\|\Omega\|\tau=\mathcal{O}(1)\) and combine it with stochastic Magnus integrators for higher-order discretization. In the fault-tolerant setting, we prove that the segmented construction supports coherent imaginary-time projection with linear scale-over-gap dependence. Unlike prior hybrid schemes, where quantum resources are used primarily to prepare improved trial states, estimate overlap data, or reduce constraint bias while the AFQMC propagation itself remains classical \cite{Huggins2022,Amsler2023_AFQMCtrial,Huang2024_QCQMC,Danilov2025_SQDphAFQMC}, we implement the  projector dynamics directly as a gate-level quantum primitive. With near-term application in mind, we formulate a linear combination of unitaries (LCU) form to trade ancilla overhead and circuit depth for additional classical sampling. We validate the method classically, verify the compiled circuits by statevector emulation, and demonstrate proof-of-principle multi-step executions on IBM quantum hardware. In the near-term sampled mode, the protocol inherits the classical AFQMC phase/sign bottleneck through estimator variance. In the coherent fault-tolerant setting, the same effect is recast as a normalization and postselection-amplitude cost for the segmented state preparation, which is amenable to future constrained-path, phaseless, and overlap-estimation based algorithms.
\section{Methods and Algorithms}

\vspace{-0.1in}

\subsection{Ground-state projection and AFQMC structure}

Consider an interacting \(N\)-fermion Hamiltonian in second quantization over an orthonormal spin-orbital basis,$\hat H \;=\; \hat H_1 + \hat H_2$, with \begin{equation} \hat H_1 = \sum_{ij} h_{ij}\, \hat c_i^\dagger \hat c_j,\; \hat H_2 = \frac{1}{2}\sum_{ijkl} V_{ijkl}\, \hat c_i^\dagger \hat c_j^\dagger \hat c_l \hat c_k . \end{equation}
Given a trial state \(\ket{\psi_T}\) with \(\braket{\psi_{\rm gs}}{\psi_T}\neq 0\), imaginary-time evolution (ITE) projects onto the ground state,
\begin{equation}\label{eq:proj}
\ket{\psi(\tau)}=e^{-\tau(\hat H-E_T)}\ket{\psi_T},
\qquad
\ket{\psi_{\rm gs}}\propto \lim_{\tau\to\infty}\ket{\psi(\tau)},
\end{equation}
where \(E_T\) is an energy shift used for stabilization. In AFQMC, \(\hat H_2\) is rewritten as a sum of squares of one-body operators,
\begin{equation}\label{eq:squares}
\hat H_2=-\frac12\sum_{\gamma=1}^m \lambda_\gamma \hat v_\gamma^2,
\end{equation}
with constants and one-body shifts absorbed into \(\hat H_1\) and \(E_T\). The HS transformation then converts the interacting projector into an average over one-body propagators driven by auxiliary fields. Since one-body propagators map Slater determinants (SDs) to SDs, this yields the structural foundation of AFQMC.

\subsection{Feynman--Kac representation and stochastic differential equations (SDEs)}

To avoid repeated kinetic-interaction splitting, we use a continuous-time stochastic representation of the projector. For clarity we first consider the case of real auxiliary fields, for which
$\hat H=\hat H_1-\sum_{\gamma=1}^m \hat L_\gamma^2,$ and $\hat L_\gamma:=\sqrt{\lambda_\gamma/2}\,\hat v_\gamma,$
so that the scalar shift \(E_T\) contributes only an overall factor \(e^{\tau E_T}\) and is omitted below. Then
\begin{equation}\label{eq:FK}
e^{-\tau \hat H}
=\displaystyle 
\mathbb{E}\!\left[
\mathcal{T}e^{
-\int_0^\tau \hat H_1\,ds
+\sum_{\gamma=1}^m \int_0^\tau \sqrt{2}\,\hat L_\gamma \circ dW_\gamma(s)
}
\right],
\end{equation}
where \(\{W_\gamma\}\) are independent Wiener processes, \(\circ\) denotes Stratonovich integration, and \(\mathbb E\) is expectation over Brownian paths. Equivalently,
\(
e^{-\tau \hat H}\ket{\psi_T}=\mathbb E[\ket{\phi_\tau}],
\)
where \(\ket{\phi_t}\) solves the Stratonovich SDE with $\ket{\phi_0}=\ket{\psi_T}$
\begin{equation}\label{eq:strat}
d\ket{\phi_t}
=
\big(
-\hat H_1\,dt
+\sum_{\gamma=1}^m \sqrt{2}\,\hat L_\gamma \circ dW_\gamma(t)
\big)\ket{\phi_t}.
\end{equation}
Our approach solves Eq.~\eqref{eq:strat} by stochastic Magnus integrators, so that each time step is represented by a single effective one-body generator. As a result, each sampled trajectory preserves SD-to-SD structure and admits a fermionic-Gaussian circuit implementation. Despite the presence of stochastic sampling, our approach is distinct from QDrift \cite{Campbell2019}, which uses randomized Hamiltonian simulation in a Trotter-style setting rather than a stochastic representation of imaginary-time projection. The general sign-indefinite Feynman--Kac representation is given in the Supplemental Material (SM).

\subsection{Stochastic Magnus expansion}


The stochastic Magnus expansion \cite{kamm2021stochastic} expresses the one-step propagator as
$U(t_{n+1},t_n)=\exp(\Omega_n)$.
Let $\Delta t=t_{n+1}-t_n$ and denote the Wiener increments
$\Delta W_{\gamma,n}=W_\gamma(t_{n+1})-W_\gamma(t_n)\sim\mathcal N(0,\Delta t)$.
For the Stratonovich SDE \eqref{eq:strat}, we write the generator in differential form as
$
A(t)\,dt \;=\; -\hat H_1\,dt \;+\; \sum_{\gamma=1}^m \sqrt{2}\,\hat L_\gamma \circ dW_\gamma(t).
$

The first-order Magnus generator is $\Omega^{(1)}_n=\int_{t_n}^{t_{n+1}}A(s)\,ds$, yielding
\begin{equation}\label{eq:weak1}
\Omega^{(1)}_n
\;=\;
-\,\Delta t\,\hat H_1
\;+\; \sum_{\gamma=1}^m \sqrt{2}\,\hat L_\gamma\,\Delta W_{\gamma,n},
\end{equation}
which, crucially, only involves one-body operators.
\vspace{-3pt}
To suppress time-step error, we retain the next term,
\begin{equation}\label{eq:weak2_general}
\Omega^{(2)}_n
=
\Omega^{(1)}_n
+
\frac{1}{2}\int_{t_n}^{t_{n+1}}\!\!\int_{t_n}^{s_1}\![A(s_1),A(s_2)]\,ds_2\,ds_1,
\end{equation} yielding an order-2 Magnus expansion.
Substituting $dA(t)$ into the commutator produces drift--noise and noise--noise contributions. For general SDEs, the noise--noise terms generate L\'evy areas and complicate higher-order schemes. Fortunately,
in the density-coupled HS decomposition, widely used in lattice-based models, the diffusion operators commute,
$[\hat L_\gamma,\hat L_{\gamma'}]=0$. The only surviving second-order
corrections are the drift--noise commutators $[\hat H_{1},\hat L_\gamma]$, giving
\begin{equation}\label{eq:Omega2}
\Omega^{(2)}_n
\;=\;
\Omega^{(1)}_n
\;-\; \sum_{\gamma=1}^m [\hat H_{1}, \hat L_\gamma]\; J_{0\gamma,n},
\end{equation}
where the Stratonovich integral $J_{0\gamma,n}$ is sampled via
\begin{equation}\label{eq:J0_def}
J_{0\gamma,n}
\;=\;
\frac{\Delta t}{2}\,\Delta W_{\gamma,n}
+\sqrt{\frac{\Delta t^3}{12}}\;\eta_{\gamma,n}, \;
\eta_{\gamma,n}\sim\mathcal N(0,1).
\end{equation}
Because the algebra of one-body operators is closed under commutators, $[\hat H_{1},\hat L_\gamma]$ remains strictly
one-body,  $\exp(\Omega^{(2)}_n)$ can again be compiled efficiently without introducing explicit two-body gates.

\subsection{Unitary dilation of  stochastic generators}

Each Magnus step produces a one-body exponent \(\Omega_k\) and a generally nonunitary slice propagator \(e^{\Omega_k}\). Over one segment of duration \(\tau=n_T\Delta t\), the AFQMC propagator is therefore
\(
B_{\mathrm{seg}}
:=
e^{\Omega_{n_T}}e^{\Omega_{n_T-1}}\cdots e^{\Omega_1}.
\)
To implement \(B_{\mathrm{seg}}\) coherently, we dilate each slice \(e^{\Omega_k}\) and compose the resulting unitaries over the segment, so that ancilla measurement is required only once per segment rather than after every time step. General dilation schemes have been developed for both deterministic and stochastic dynamics \cite{jin2022quantum,jin2025quantum}. Here we follow the moment-matching dilation framework \cite{li2025linear,wu2026universal}, and we split each slice exponent as
\(
\Omega_k=K_k-iH_k,
\)
with \(H_k\) and \(K_k\) Hermitian, and introduce an ancilla space \(\mathcal H_A\) with boundary states \(\langle l|\), \(|r\rangle\) and a skew-Hermitian generator \(F\) (\(F^\dagger=-F\)). The dilated slice Hamiltonian is
\begin{equation}\label{eq:dilatedH}
\widetilde H_k
=
\mathbb I_A\otimes H_k+(iF)\otimes K_k .
\end{equation}
which is Hermitian. The corresponding segment unitary is
\(
U_{\mathrm{seg}}
:=
e^{-i\widetilde H_{n_T}}e^{-i\widetilde H_{n_T-1}}\cdots e^{-i\widetilde H_1}.
\)
Projecting the ancilla, we get
\(
B_{\mathrm{seg}}
\approx
(\langle l|\otimes \mathbb I_S)\,U_{\mathrm{seg}}\,(|r\rangle\otimes \mathbb I_S),
\)
with a controllable postselection amplitude.

We choose \(iF\) as the nearest-neighbor hopping Hamiltonian on a 1D tight-binding chain \cite{li2025linear}. For an \(n_A\)-qubit ancilla (dimension \(d_A=2^{n_A}\)), the effective moment order scales as \(m=\Theta(d_A)\), and the dilation error  obeys
\[
\Big\|e^{\Omega_k}-(\langle l|\otimes \mathbb I)e^{-i\widetilde H_k}(|r\rangle\otimes \mathbb I)\Big\|
=
\mathcal O\!\left(\frac{\|\Omega_k\|^m}{m!}\right).
\]
Thus increasing \(n_A\) suppresses the dilation error super-polynomially. By contrast, first-order single-qubit-ancilla constructions based on \(\sigma_y\)-controlled primitives \cite{liu2021probabilistic,Mao2023MeasurementBasedDeterministicITE} are intrinsically small-step, with per-slice error \(\mathcal O(\|\Omega_n\|^2 \dt^2 ) \), and therefore require \(\|\Omega_n\| \dt \ll 1\).
By approximating \(\Delta W\) with a discrete distribution \cite{kloeden1992numerical}, we obtain bounded  slices and hence a uniform bound
\(
\|\Omega_k\|\le \|\Omega\|\,\Delta t
\)
for all admissible slices, where \(\|\Omega\|\) is independent of \(k\) (see the SM). Under this bound, the postselection success remains \(\Omega(1)\) provided \(\|\Omega\|\tau=\co(1)\). 
 
\vspace{-0.05in}
\begin{theorem}
\label{thm:main}
Let \(\hat H\) be Hermitian with ground energy \(E_0\) and gap \(\Delta>0\), and let \(\ket{\psi_T}\) have nonzero ground overlap $\gamma>0$.
Choose \(\tau=\Theta(1/\norm{\Omega})\) with \(\tau=n_T\Delta t\) and
\(\dt=\Theta(\sqrt{\epsilon_{\rm prep}})\).
If the per-segment normalized preparation error is sufficiently small, namely
\[
\epsilon_{\rm prep}
=
\mathcal O\!\left(
\frac{(1-e^{-2\tau\Delta})\,\epsilon}{\|\hat H\|}
\right),
\]
then after
\(
S=\Theta\!\Big(\frac{\norm{\Omega}}{\Delta}\log\frac{1}{\gamma\epsilon}\Big)
\)
segments one obtains a state \(\ket{\psi_S}\) with
\(
\bra{\psi_S}\hat H\ket{\psi_S}-E_0\le \epsilon.
\)
Given a block-encoding of \(\hat H/\alpha\), the final energy can then be estimated to additive error \(\epsilon\) with failure probability at most \(\delta\) using
\(
\widetilde{\mathcal O}\!\big(\alpha\,\epsilon^{-1}\log(1/\delta)\big)
\)
coherent queries \cite{Rall2021}.
\end{theorem}
\vspace{-0.05in}
The theorem yields a coherent Feynman--Kac ground-state preparation algorithm with the same characteristic linear scale-over-gap dependence as spectral-filtering methods, but built from one-body segment primitives rather than a global block-encoding input oracle \cite{Berry2015TaylorLCU,LinTong2020nearoptimalground,DongLinTong2022QETU}. In addition, the classical sign problem is recast into the normalization of the projected state and the associated postselection amplitude, suggesting that coherent overlap estimation and amplitude-amplification based variants will soften the long-time sampling bottleneck.

The fault-tolerant construction assumes coherent segment chaining, ancilla restoration, oblivious amplitude amplification, and controlled unitaries for coherent energy estimation and trajectory superposition, which are beyond current near-term capabilities. We therefore adopt an implementable mode that keeps the same slice kernel but samples the Magnus randomness classically, while the circuit layer evaluates a small deterministic LCU over retained system-only terms. This removes coherent trajectory superposition and persistent ancilla management, yielding shallower circuits suitable for present hardware.

\subsection{System-only implementation via LCU}\label{sec:lcu} 

\begin{figure}[htp] 
\centering 
\includegraphics[scale=0.55]{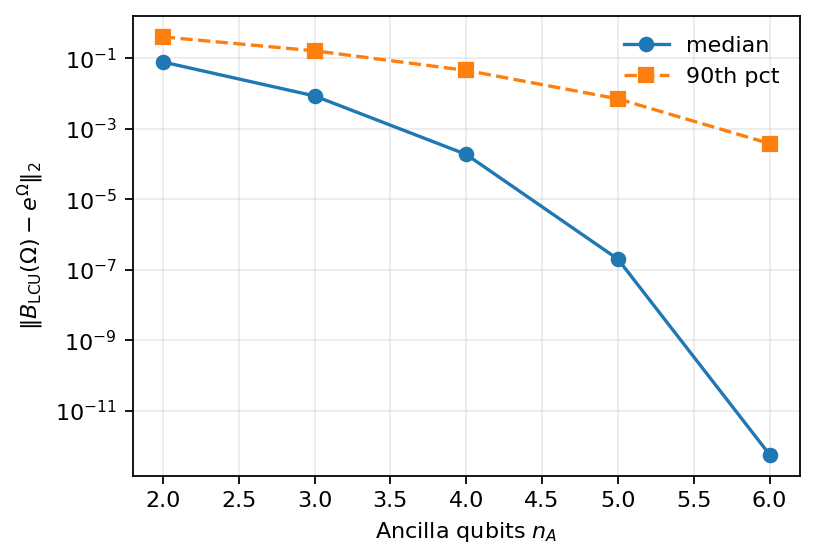} 
\vspace{-6pt}
\caption{One-step LCU  error 
versus ancilla size \(n_A\)  over 100 sampled order-2 Magnus slices \(\Omega_k\). The median and 90th percentile both decrease rapidly with \(n_A\), showing a modest ancilla is sufficient to make the approximation accurate.} 
\vspace{-6pt}
\label{fig:lcu-error-vs-nA} 
\end{figure}
The fault-tolerant construction assumes coherent segment chaining, ancilla restoration, oblivious amplitude amplification, and controlled unitaries for coherent energy estimation and trajectory superposition, which are beyond current hardware capabilities. 
One key observation is that the dilation step can be rewritten as a short linear combination of system-only unitaries by diagonalizing the ancilla generator \(iF\). Writing \( iF=\sum_{j=1}^{d_A}\omega_j\ketbra{\chi_j}, \, \omega_j\in\mathbb R, \) each slice is reduced to $U_{n,j}=e^{-i(H_n+\omega_j K_n)}$ and $
B_{\mathrm seg}=\sum_{j=1}^{d_A} c_j\,U_{n,j}U_{n-1,j}\cdots U_{1,j},  
$
where \(c_j=\langle l|\chi_j\rangle\langle\chi_j|r\rangle\). 
Thus the ancilla-assisted unitary dilation of a nonunitary segment is replaced by a deterministic LCU over the retained indices \(j\), and each term \(U_{k,j}\) acts only on the system register. Since \(U_{k,j}\) is generated by a quadratic fermionic Hamiltonian, it can be compiled as a fermionic Gaussian circuit, rather than through an additional Trotter decomposition, substantially reducing the two-qubit depth relative to a direct circuit for the  dilated Hamiltonian \(\widetilde H\). For near-term implementations, the Magnus randomness is sampled classically as in AFQMC, while the LCU over \(j\) is evaluated deterministically. This removes the dilation register from the executed circuit and trades ancilla overhead for additional classical preprocessing and a moderate increase in the number of measured system-only circuits.   This passage from dilation to an LCU of system-only Hamiltonians is analogous to the relation between Schr\"odingerization \cite{jin2022quantum} and LCHS \cite{ALL23}.

In the near-term mode, the sampling overhead grows linearly with the number of classical trajectories \(N_{\rm traj}\), the number \(r\) of LCU terms, and the number \(N_{\rm obs}\) of observables entering the mixed estimator. For a fixed shot budget per circuit, this gives a total measurement cost proportional to \(N_{\rm traj} r N_{\rm obs}\). For a system with \(n_{\rm orb}\) spin-orbitals, each retained segment unitary is a quadratic-fermion propagator and can be compiled directly as a fermionic Gaussian circuit using \(\mathcal O(n_{\rm orb}^2)\) Givens rotations, avoiding Trotter decomposition and helping keep the coherent depth low on near-term hardware. 
\section{Numerical test}

We benchmark the stochastic Magnus+dilation/LCU construction on the half-filled two-site Hubbard dimer in a four-spin-orbital Jordan--Wigner encoding, initialized from the Hartree--Fock determinant \(|1001\rangle\). The code uses \(t=1\), \(U=4\), and the particle-hole-shifted interaction \(U\sum_i (n_{i\uparrow}-\tfrac12)(n_{i\downarrow}-\tfrac12)\). The Hubbard dimer is used here as a minimal but fully controlled benchmark for validating the complete AFQMC pipeline, rather than as a claim of scale or advantage.  Classical checks use sampled AFQMC trajectories, while the hardware study uses mixed estimators built from compiled circuits. To reduce circuit depth, the circuit synthesis uses the reordered orbital basis \((1\uparrow,2\uparrow,1\downarrow,2\downarrow)\), for which the one-body propagator is spin-block diagonal.
We first examine the LCU kernel itself. Figure~\ref{fig:lcu-error-vs-nA} shows the one-step LCU approximation error as a function of ancilla size. Over sampled order-2 Magnus slices, the median error drops rapidly with \(n_A\),  while the 90th-percentile error decreases  substantially as well. This rapid convergence shows that only a modest ancilla is needed to control the per-segment dilation error in the constant-success regime. 

We next validate the stochastic Magnus propagation at the classical level. Figure~\ref{fig:w1w2} compares AFQMC energies obtained from unitary dilations of the first- and second-order Magnus discretizations against deterministic imaginary-time references from \(\exp(-\tau \hat H)\). For $\dt=0.05$, the second-order scheme follows both the mixed-estimator and Rayleigh-quotient references more closely over \(\tau\le 0.3\), while the first-order scheme shows a visibly larger time-discretization bias. Thus the pathwise propagation captures the intended imaginary-time relaxation, and the weak-order improvement translates directly into more accurate physical observables.

\begin{figure}[htp]
    \centering
    \includegraphics[scale=0.4]{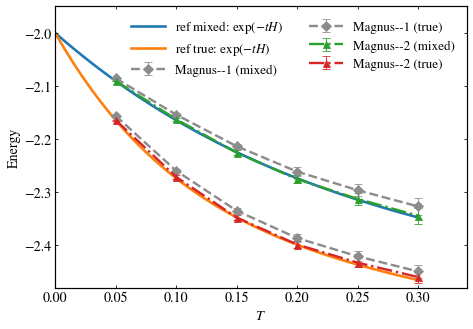}
    \caption{Energy estimates versus imaginary time \(T\) comparing Magnus--1 and Magnus--2 AFQMC propagation. Dashed and dash--dot curves with markers show mixed and true (Rayleigh) estimators with jackknife error bars.  
    }
    \label{fig:w1w2}
\end{figure}

We then test segmented propagation. Mirroring the fault-tolerant construction in Theorem~\ref{thm:main}: each segment prepares a normalized approximation to \(\ket{\psi(\tau)}\), and the next segment proceeds with fresh Magnus slices.  Figure~\ref{fig:w2seg2} compares the two-segment ME-2 dilation simulation with the deterministic imaginary-time reference. The mixed estimator closely follows the corresponding mixed-reference curve, while the exact estimator remains consistent with the Rayleigh quotient, showing that the propagated ensemble continues to approximate the intended projected state after segment chaining. 

\begin{figure}[htp]
    \centering
    \includegraphics[scale=0.45]{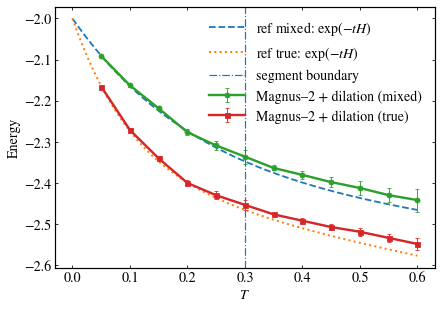}
    \caption{Energy estimates versus imaginary time \(\tau\) for the two-segment Magnus--2 dilation simulation. Circles (squares) show the mixed (true) estimator with path-sampling error bars; dashed (dotted) curves denote the corresponding deterministic references from \(\exp(-\tau\hat H)\). The vertical dash--dot line marks the segment boundary.
    }
    \label{fig:w2seg2}
\end{figure}

Finally, we assess the near-term viability of the implementation on IBM-Boston. We use the trial state \(\ket{\psi_T}\) as the reference state in the mixed estimator, sample \(N_{\mathrm{traj}}=40\) trajectories, and use 400 shots per circuit. For each trajectory, the AFQMC segment map is compiled as a fermionic Gaussian circuit via Givens rotations. The hardware runs use IBM Runtime error-suppression and mitigation options, including resilience level 1 and dynamical decoupling. The resulting circuits use 5 logical qubits and transpile to 7 active physical qubits, with depth 271 and two-qubit count 77 CZ gates. The emulator follows the exact trend qualitatively, and the IBM results remain in the same energy window, although with larger variance. 

\begin{figure}[htp]
    \centering
    \includegraphics[scale=0.4]{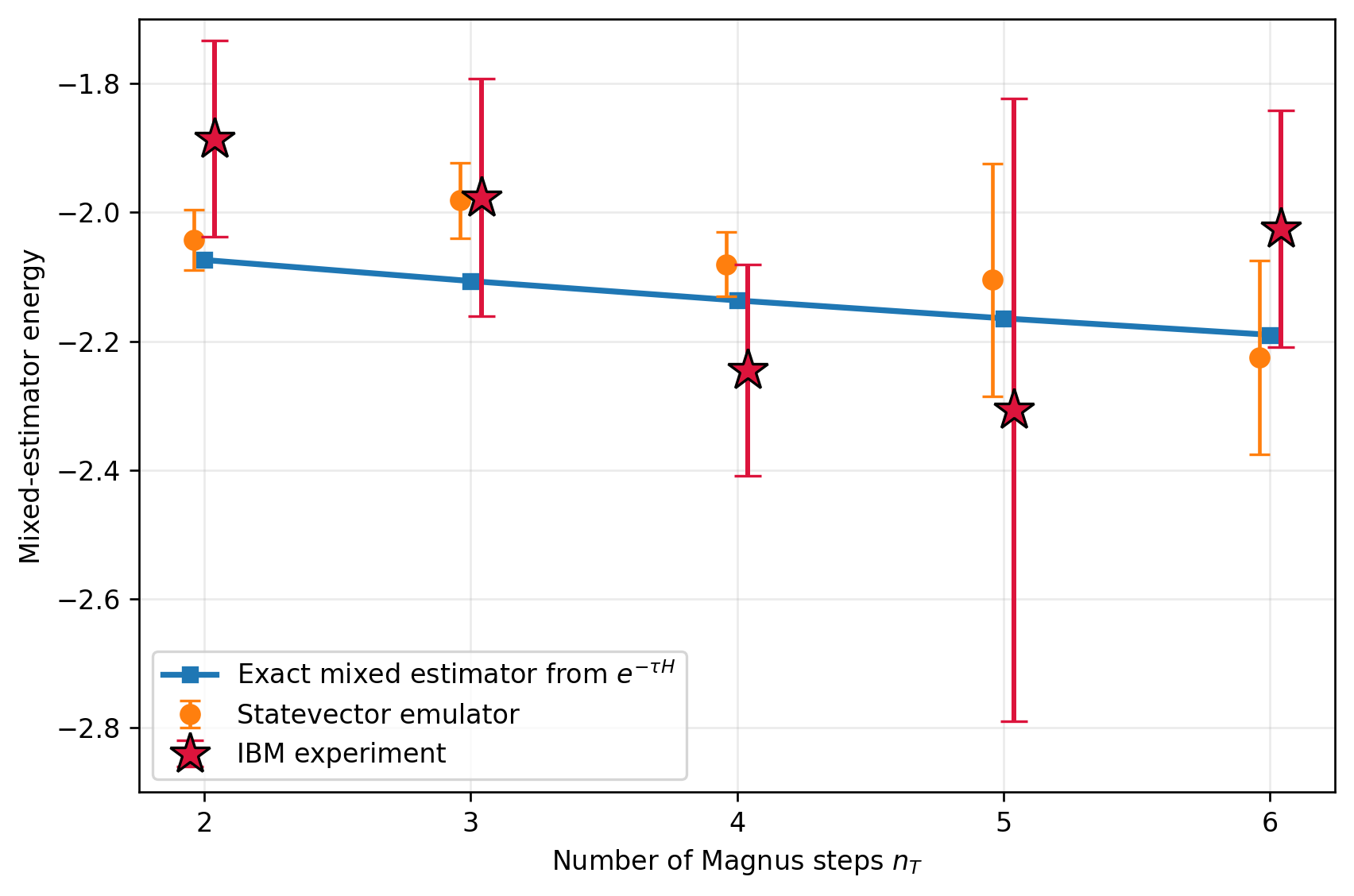}
    \caption{Mixed-estimator energies versus the number of Magnus steps \(n_T\). Blue squares: exact mixed estimator. Orange circles: statevector emulation of the compiled  circuits. Red stars: IBM results. Error bars denote bootstrap errors.}
    \label{fig:ibm}
\end{figure}

\section{Summary and Discussion}


We introduced a quantum-native AFQMC framework in which the Feynman-Kac projector dynamics is discretized into high-order stochastic Magnus slices and each nonunitary one-body step is implemented by dilation or LCU, preserving Slater-to-Slater structure and yielding hardware-executable fermionic-Gaussian circuits. While the framework does not by itself remove the classical sign bottleneck in the sampled near-term regime, it provides a gate-level realization of projector dynamics and a concrete route for benchmarking on quantum hardware. The sign problem appears directly through the success probability and estimator variance, suggesting extensions based on phaseless controls and overlap estimation \cite{TroyerWiese2005,ZhangKrakauer2003_Phaseless,ShiZhang2021_AFQMC_Review}. The same framework may be extended to finite-temperature simulation and connects naturally to recent Gaussian/Lindblad state-preparation ideas \cite{Motta2020NatPhys,Sun2021,LiZhanLin2024_DissipativeGSP,Zhan2025_RapidDissipativeGSP}.
\emph{Acknowledgements.--}
This work was supported by NSF Grant DMS-2411120. This research used resources of the National Energy Research Scientific Computing Center (NERSC) under Contract No.~DE-AC02-05CH11231 and award DDR-ERCAP0035683.
\bibliographystyle{apsrev4-2}
\bibliography{ref}

@article{Hubbard1959,
  author    = {J. Hubbard},
  title     = {Calculation of Partition Functions},
  journal   = {Phys. Rev. Lett.},
  volume    = {3},
  pages     = {77},
  year      = {1959}
}

@article{Stratonovich1958,
  author    = {R. L. Stratonovich},
  title     = {On a Method of Calculating Quantum Distribution Functions},
  journal   = {Sov. Phys. Dokl.},
  volume    = {2},
  pages     = {416},
  year      = {1958}
}

@article{liu2021probabilistic,
  title={Probabilistic nonunitary gate in imaginary time evolution},
  author={Liu, Tong and Liu, Jin-Guo and Fan, Heng},
  journal={Quantum Information Processing},
  volume={20},
  number={6},
  pages={204},
  year={2021},
  publisher={Springer}
}

@article{LinTong2020nearoptimalground,
  title   = {Near-optimal ground state preparation},
  author  = {Lin, Lin and Tong, Yu},
  journal = {Quantum},
  volume  = {4},
  pages   = {372},
  year    = {2020},
  doi     = {10.22331/q-2020-12-14-372},
  eprint  = {2002.12508},
  archivePrefix = {arXiv},
  primaryClass  = {quant-ph}
}

@article{DongLinTong2022QETU,
  title   = {Ground state preparation and energy estimation on early fault-tolerant quantum computers via quantum eigenvalue transformation of unitary matrices},
  author  = {Dong, Yulong and Lin, Lin and Tong, Yu},
  journal = {PRX Quantum},
  volume  = {3},
  pages   = {040305},
  year    = {2022},
  doi     = {10.1103/PRXQuantum.3.040305},
  eprint  = {2204.05955},
  archivePrefix = {arXiv},
  primaryClass  = {quant-ph}
}

@article{kamm2021stochastic,
  title={On the stochastic Magnus expansion and its application to SPDEs},
  author={Kamm, Kevin and Pagliarani, Stefano and Pascucci, Andrea},
  journal={Journal of Scientific Computing},
  volume={89},
  number={3},
  pages={56},
  year={2021},
  publisher={Springer}
}

@article{TroyerWiese2005,
  author    = {M. Troyer and U.-J. Wiese},
  title     = {Computational Complexity and Fundamental Limitations to Fermionic Quantum Monte Carlo Simulations},
  journal   = {Phys. Rev. Lett.},
  volume    = {94},
  pages     = {170201},
  year      = {2005}
}

@book{kloeden1992numerical,
  author    = {Peter E. Kloeden and Eckhard Platen},
  title     = {Numerical solution of stochastic differential equations},
  publisher = {Springer},
  year      = {1992}
}

@article{jin2022quantum,
  author    = {Shi Jin and Nana Liu and Yue Yu},
  title     = {Quantum simulation of partial differential equations via {Schr\"odingerization}},
  journal   = {Phys. Rev. Lett.},
  volume    = {133},
  pages     = {230602},
  year      = {2024}
}

@misc{li2025linear,
  author = {Xiantao Li},
  title = {From Linear Differential Equations to Unitaries: A Moment-Matching Dilation Framework with Near-Optimal Quantum Algorithms},
  year = {2025},
  eprint = {2507.10285},
  archivePrefix = {arXiv}
}

@article{Rall2021,
  title = {Faster coherent quantum algorithms for phase, energy, and amplitude estimation},
  author = {Rall, Patrick},
  journal = {Quantum},
  volume = {5},
  pages = {491},
  year = {2021},
  month = {jul},
  publisher = {Verein zur F{\"{o}}rderung des Open Access Publizierens in den Quantenwissenschaften},
  doi = {10.22331/q-2021-07-08-491},
  url = {https://doi.org/10.22331/q-2021-07-08-491}
}

@article{McClean2016_VQETheory,
  author  = {McClean, Jarrod R. and Romero, Jonathan and Babbush, Ryan and Aspuru-Guzik, Al{\'a}n},
  title   = {The theory of variational hybrid quantum-classical algorithms},
  journal = {New Journal of Physics},
  volume  = {18},
  number  = {2},
  pages   = {023023},
  year    = {2016},
  doi     = {10.1088/1367-2630/18/2/023023},
  eprint  = {1509.04279},
  archivePrefix = {arXiv},
  primaryClass  = {quant-ph}
}

@article{Campbell2019,
  title = {Random Compiler for Fast Hamiltonian Simulation},
  author = {Campbell, Earl},
  journal = {Phys. Rev. Lett.},
  volume = {123},
  issue = {7},
  pages = {070503},
  numpages = {6},
  year = {2019},
  month = {Aug},
  publisher = {American Physical Society},
  doi = {10.1103/PhysRevLett.123.070503}
}

@article{LiZhanLin2024_DissipativeGSP, title = {Dissipative ground state preparation in ab initio electronic structure theory}, author = {Li, Hao-En and Zhan, Yongtao and Lin, Lin}, journal = {arXiv preprint arXiv:2411.01470}, year = {2024}, eprint = {2411.01470}, archivePrefix = {arXiv}, primaryClass = {quant-ph}, url = {https://arxiv.org/abs/2411.01470} }

@article{Zhan2025_RapidDissipativeGSP, title = {Rapid quantum ground state preparation via dissipative dynamics}, author = {Zhan, Yongtao and Ding, Zhiyan and Huhn, Jakob and Gray, Johnnie and Preskill, John and Chan, Garnet Kin-Lic and Lin, Lin}, journal = {arXiv preprint arXiv:2503.15827}, year = {2025}, eprint = {2503.15827}, archivePrefix = {arXiv}, primaryClass = {quant-ph}, url = {https://arxiv.org/abs/2503.15827} }

@article{Jiang2018_FermionicGaussian,
  author  = {Jiang, Zhang and Sung, Kevin J. and Kechedzhi, Kostyantyn and Smelyanskiy, Vadim N. and Boixo, Sergio},
  title   = {Quantum Algorithms to Simulate Many-Body Physics of Correlated Fermions},
  journal = {Physical Review Applied},
  volume  = {9},
  number  = {4},
  pages   = {044036},
  year    = {2018},
  doi     = {10.1103/PhysRevApplied.9.044036},
  eprint  = {1711.05395},
  archivePrefix = {arXiv},
  primaryClass  = {quant-ph}
}

@article{Kivlichan2018_LinearDepthSlater,
  author  = {Kivlichan, Ian D. and McClean, Jarrod and Wiebe, Nathan and Gidney, Craig and Aspuru-Guzik, Al{\'a}n and Chan, Garnet Kin-Lic and Babbush, Ryan},
  title   = {Quantum Simulation of Electronic Structure with Linear Depth and Connectivity},
  journal = {Physical Review Letters},
  volume  = {120},
  number  = {11},
  pages   = {110501},
  year    = {2018},
  doi     = {10.1103/PhysRevLett.120.110501},
  eprint  = {1711.04789},
  archivePrefix = {arXiv},
  primaryClass  = {quant-ph}
}

@article{Danilov2025_SQDphAFQMC,
  title         = {Enhancing the accuracy and efficiency of sample-based quantum diagonalization with phaseless auxiliary-field quantum Monte Carlo},
  author        = {Danilov, Don and Robledo-Moreno, Javier and Sung, Kevin J. and Motta, Mario and Shee, James},
  journal       = {arXiv preprint arXiv:2503.05967},
  year          = {2025},
  eprint        = {2503.05967},
  archivePrefix = {arXiv},
  primaryClass  = {quant-ph},
  url           = {https://arxiv.org/abs/2503.05967}
}

@article{ALL23,
  author    = {Dong An and Jin-Peng Liu and Lin Lin},
  title     = {Linear combination of {H}amiltonian simulation for nonunitary dynamics with optimal state preparation cost},
  journal   = {Physical Review Letters},
  volume    = {131},
  number    = {15},
  pages     = {150603},
  year      = {2023},
  note      = {arXiv:2303.01029}
}

@article{jin2025quantum,
  author    = {Shi Jin and Nana Liu and Wei Wei},
  title     = {Quantum algorithms for stochastic differential equations: A schr{\"o}dingerisation approach},
  journal   = {Journal of Scientific Computing},
  volume    = {104},
  number    = {2},
  pages     = {1--32},
  year      = {2025}
}

@article{Sun2021,
  author    = {S.-N. Sun and M. Motta and R. N. Tazhigulov and A. T. K. Tan and G. K.-L. Chan and A. J. Minnich},
  title     = {Quantum Imaginary Time Evolution with a Combined Variational and Trotter Decomposition},
  journal   = {PRX Quantum},
  volume    = {2},
  pages     = {010317},
  year      = {2021}
}

@article{Huggins2022,
  author    = {W. J. Huggins and B. A. O'Gorman and N. C. Rubin and D. R. Reichman and R. Babbush and J. Lee},
  title     = {Unbiasing fermionic quantum Monte Carlo with a quantum computer},
  journal   = {Nature},
  volume    = {603},
  pages     = {416--420},
  year      = {2022}
}

@article{McArdle2019,
  author    = {S. McArdle and T. Jones and S. Endo and Y. Li and S. C. Benjamin and X. Yuan},
  title     = {Variational ansatz-based quantum simulation of imaginary time evolution},
  journal   = {npj Quantum Inf.},
  volume    = {5},
  pages     = {75},
  year      = {2019}
}

@article{Mao2023MeasurementBasedDeterministicITE,
  title   = {Measurement-Based Deterministic Imaginary Time Evolution},
  author  = {Mao, Yuping and Chaudhary, Manish and Kondappan, Manikandan and Shi, Junheng and Ilo-Okeke, Ebubechukwu O. and Ivannikov, Valentin and Byrnes, Tim},
  journal = {Physical Review Letters},
  volume  = {131},
  number  = {11},
  pages   = {110602},
  year    = {2023},
  doi     = {10.1103/PhysRevLett.131.110602}
}

@article{Motta2020NatPhys,
  author  = {Motta, Mario and Sun, Chong and Tan, Adrian T. K. and O'Rourke, Matthew J. and Ye, Erika and Minnich, Austin J. and Brand{\~a}o, Fernando G. S. L. and Chan, Garnet Kin-Lic},
  title   = {Determining eigenstates and thermal states on a quantum computer using quantum imaginary time evolution},
  journal = {Nature Physics},
  year    = {2020},
  volume  = {16},
  pages   = {205--210},
  doi     = {10.1038/s41567-019-0704-4}
}

@article{Lee2021PRXQTHC,
  author  = {Lee, Joonho and Berry, Dominic W. and Gidney, Craig and Motta, Mario and McClean, Jarrod R. and Babbush, Ryan},
  title   = {Even More Efficient Quantum Computations of Chemistry Through Tensor Hypercontraction},
  journal = {PRX Quantum},
  year    = {2021},
  volume  = {2},
  number  = {3},
  pages   = {030305},
  doi     = {10.1103/PRXQuantum.2.030305}
}

@article{Berry2015TaylorLCU,
  author  = {Berry, Dominic W. and Childs, Andrew M. and Cleve, Richard and Kothari, Robin and Somma, Rolando D.},
  title   = {Simulating Hamiltonian Dynamics with a Truncated Taylor Series},
  journal = {Physical Review Letters},
  year    = {2015},
  volume  = {114},
  number  = {9},
  pages   = {090502},
  doi     = {10.1103/PhysRevLett.114.090502}
}

@article{Hirsch1983_AFQMC,
  title = {Discrete Hubbard-Stratonovich transformation for fermion lattice systems},
  author = {Hirsch, J. E.},
  journal = {Phys. Rev. B},
  volume = {28},
  number = {7},
  pages = {4059--4061},
  year = {1983},
  doi = {10.1103/PhysRevB.28.4059}
}

@article{ZhangCarlsonGubernatis1997_CPMC,
  title = {Constrained path Monte Carlo method for fermion ground states},
  author = {Zhang, Shiwei and Carlson, J. and Gubernatis, J. E.},
  journal = {Phys. Rev. B},
  volume = {55},
  number = {12},
  pages = {7464--7477},
  year = {1997},
  doi = {10.1103/PhysRevB.55.7464}
}

@article{ZhangKrakauer2003_Phaseless,
  title = {Quantum Monte Carlo Method Using Phase-Free Random Walks with Slater Determinants},
  author = {Zhang, Shiwei and Krakauer, Henry},
  journal = {Phys. Rev. Lett.},
  volume = {90},
  pages = {136401},
  year = {2003},
  doi = {10.1103/PhysRevLett.90.136401}
}

@article{MottaZhangChan2019_AFQMCsym,
  author  = {Motta, Mario and Zhang, Shiwei and Chan, Garnet Kin-Lic},
  title   = {Hamiltonian symmetries in auxiliary-field quantum Monte Carlo calculations for electronic structure},
  journal = {Physical Review B},
  volume  = {100},
  number  = {4},
  pages   = {045127},
  year    = {2019},
  doi     = {10.1103/PhysRevB.100.045127},
  eprint  = {1905.00511},
  archivePrefix = {arXiv},
  primaryClass  = {physics.comp-ph}
}

@article{Amsler2023_AFQMCtrial,
  author  = {Amsler, Michael and Deglmann, Peter and Degroote, Matthias and
             Kaicher, Matthias P. and Kiser, Markus and K{\"u}hn, Martin and
             Kumar, Chetan and Maier, Andreas and Samsonidze, George and
             Schroeder, Andreas and Streif, Michael and Vodola, Davide and
             Wever, Christoph},
  title   = {Classical and quantum trial wave functions in auxiliary-field quantum Monte Carlo applied to oxygen allotropes and a CuBr$_2$ model system},
  journal = {The Journal of Chemical Physics},
  volume  = {159},
  number  = {4},
  pages   = {044119},
  year    = {2023},
  doi     = {10.1063/5.0146934}
}

@article{Huang2024_QCQMC,
  author  = {Huang, Benchen and Chen, Yi-Ting and Gupt, Brajesh and
             Suchara, Martin and Tran, Anh and McArdle, Sam and Galli, Giulia},
  title   = {Evaluating a quantum-classical quantum Monte Carlo algorithm with Matchgate shadows},
  journal = {Physical Review Research},
  volume  = {6},
  number  = {4},
  pages   = {043063},
  year    = {2024},
  doi     = {10.1103/PhysRevResearch.6.043063}
}

@article{ShiZhang2021_AFQMC_Review,
  title = {Auxiliary-field quantum Monte Carlo: An overview},
  author = {Shi, Hao and Zhang, Shiwei},
  journal = {J. Chem. Phys.},
  volume = {154},
  number = {2},
  pages = {024107},
  year = {2021},
  doi = {10.1063/5.0036522}
}

@article{Peruzzo2014_VQE,
  title = {A variational eigenvalue solver on a photonic quantum processor},
  author = {Peruzzo, A. and McClean, J. and Shadbolt, P. and Yung, M.-H. and Zhou, X.-Q. and Love, P. J. and Aspuru-Guzik, A. and O'Brien, J. L.},
  journal = {Nat. Commun.},
  volume = {5},
  pages = {4213},
  year = {2014},
  doi = {10.1038/ncomms5213}
}

@article{Kandala2017_IBM_HEVQE,
  title = {Hardware-efficient variational quantum eigensolver for small molecules and quantum magnets},
  author = {Kandala, Abhinav and Mezzacapo, Antonio and Temme, Kristan and Takita, Maika and Brink, Markus and Chow, Jerry M. and Gambetta, Jay M.},
  journal = {Nature},
  volume = {549},
  pages = {242--246},
  year = {2017},
  doi = {10.1038/nature23879}
}

@article{Babbush2018_PRX_Qubitization,
  title = {Encoding Electronic Spectra in Quantum Circuits with Linear {T} Complexity},
  author = {Babbush, Ryan and Gidney, Craig and Berry, Dominic W. and Wiebe, Nathan and McClean, Jarrod and Paler, Alexandru and Fowler, Austin and Neven, Hartmut},
  journal = {Phys. Rev. X},
  volume = {8},
  pages = {041015},
  year = {2018},
  doi = {10.1103/PhysRevX.8.041015}
}

@misc{wu2026universal,
  title = {Universal Dilation of Linear It{\^o} {SDE}s: Quantum Trajectories and Lindblad Simulation of Second Moments},
  author = {Wu, Hsuan-Cheng and Li, Xiantao},
  year = {2026},
  eprint = {2601.05928},
  archivePrefix = {arXiv},
  primaryClass = {quant-ph}
}

@article{BauerBravyiMottaChan2020_ChemRev,
  title   = {Quantum Algorithms for Quantum Chemistry and Quantum Materials Science},
  author  = {Bauer, Bela and Bravyi, Sergey and Motta, Mario and Chan, Garnet Kin-Lic},
  journal = {Chem. Rev.},
  volume  = {120},
  number  = {22},
  pages   = {12685--12717},
  year    = {2020},
  doi     = {10.1021/acs.chemrev.9b00829}
}

\onecolumngrid
\appendix

\newpage 

\begin{center}
   \bf \large Auxiliary-Field Quantum Monte Carlo on Quantum Hardware via Unitary Dilation.
\end{center}

\section{The moment-fulfilling dilation method}
Suppose we are simulating a non-unitary evolution with generator $\Omega$ with both $H$ and $K$ being Hermitian. 
Our moment-fulfilling framework is built on the projected dilated propagator
\begin{equation}\label{dil}
(\bra{l}\otimes I)\,e^{-it\widetilde H}\,(\ket{r}\otimes I),
\qquad
\widetilde H := I\otimes H + (iF)\otimes K,
\end{equation}
which recovers the target non-unitary evolution exactly when the moment identities
\[
\langle l|F^k|r\rangle =1,\qquad \forall k\ge 0,
\]
hold. The proof is algebraic \cite{li2025linear}: expand the unitary $e^{-it\widetilde H}$ (or its time-ordered Dyson series) and use the
moment identities to collapse every occurrence of $F^k$ inside the matrix element, thereby reproducing the Taylor/Dyson
series of $e^{t(-iH+K)}$ term-by-term.

In the finite-dimensional tight-binding realization used here, $F$ is approximated by the skew-Hermitian chain
$\theta F_h$ on grid points $p_j=jh$ with $h=1/M$ and the trapezoid-weight. When $\theta=2$, we have 
\begin{equation}\label{Fh}
  \begin{aligned}
    (F_h \bm v)_0 &= \tfrac{1}{2\sqrt2}  v_{1},\\
    (F_h \bm v)_1 &= \tfrac34  v_{2}-\tfrac1{2\sqrt2}  v_{0},\\
    (F_h \bm v)_i &= \tfrac{p_{i+1}+p_i}{4h}  v_{i+1}
              - \tfrac{p_i+p_{i-1}}{4h}  v_{i-1},
              \quad 2\le i\le M-1,\\
    (F_h \bm v)_M &= -\frac{p_{M-1}+p_M}{4h}  v_{M-1}.
  \end{aligned}
\end{equation}
In addition, the vector $\ket{r_h}$ is given by,
\begin{equation}\label{rh}
\ket{r_h} = \frac{h}{2}\ket{0}+ h \sum_{i=1}^{M-1} \ket{i} + \frac{h}{2}\ket{M}.
\end{equation}
When implemented with an ancilla register, we choose $n_A= \log_2 (M+1)$. 

A key property is that $\ket{r_h}$ is an \emph{interior} eigenmode of $\theta F_h$ up to a rank-one defect at the
right boundary:
\begin{equation}\label{eq:boundary_defect}
(\theta F_h)\ket{r_h}-\ket{r_h} = \alpha\,\ket{M},
\end{equation}
for some scalar $\alpha$ \cite{li2025linear}. Equivalently, the residual vanishes on all interior nodes $j<M$ and is supported only at the terminal site $\ket{M}$.

The dilated evolution in \eqref{dil} will be followed by post-selection to reduce the system register to follow the non-unitary evolution. 
Let $\Pi_{\rm win}$ denote projection onto a prefront window that excludes the boundary, e.g.
\[
\Pi_{\rm win}=\sum_{j=0}^{j^\ast}\ketbra{j},\qquad j^\ast = \frac{M}2.
\]
Because $F_h$ is nearest-neighbor on the ancilla chain, any component injected at $\ket{M}$ can propagate left by at
most one site per application of $F_h$. Consequently, for every $k\le m:=M-j^\ast$ the boundary defect in
\eqref{eq:boundary_defect} cannot reach the window, and we have the \emph{windowed moment identities}
\begin{equation}\label{eq:window_moments}
\Pi_{\rm win}(\theta F_h)^k\ket{r_h}=\Pi_{\rm win}\ket{r_h},
\qquad\text{hence}\qquad
\langle l_h|(\theta F_h)^k|r_h\rangle =1,\quad k=0,1,\dots,m,
\end{equation}
where $\langle l_h|$ is any evaluation functional supported inside the window,
\begin{equation}\label{lh}
    \bra{l_h} = \frac{\sum_{i=0}^{j_\ast} \bra{i}  }{\sum_{i=0}^{j_\ast} \braket{i}{r_h} }. 
\end{equation}

Together, \cref{Fh,rh,lh} constitute the elements needed to apply the dilation scheme \eqref{dil}.

\bigskip 
Define the target slice map $B(t):=e^{t(-iH+K)}$, and let
$\widetilde B(t)$ be the windowed dilation kernel obtained by applying $(l_h|\otimes I)$ and $(|r_h)\otimes I)$ together
with the window projection $\Pi_{\rm win}$.
Expanding both $B(t)$ and $\widetilde B(t)$ in Taylor series, the finite moment identities
\eqref{eq:window_moments} imply that the two series \emph{agree through order $m$}. Therefore the mismatch is controlled by
the Taylor remainder:
\begin{equation}\label{eq:trunc_remainder}
\|B(t)-\widetilde B(t)\|
\;\le\;
\sum_{k=m+1}^{\infty}\frac{(t\|\Omega\|)^k}{k!}
\;\le\;
\frac{(t\|\Omega\|)^{m+1}}{(m+1)!}\,e^{t\|\Omega\|},
\end{equation}
and in the constant-success regime where $t\|\Omega\|=\co(1)$ this simplifies to the advertised factorial scaling
\[
\|B(t)-\widetilde B(t)\| \;\lesssim\; \frac{(t\|\Omega\|)^{m+1}}{(m+1)!}.
\]
With the symmetric choice $j^\ast=M/2$ (so $m=M/2$), we obtain
\[
\|B(t)-\widetilde B(t)\|
\;\lesssim\;
\frac{(t\|\Omega\|)^{M/2+1}}{(M/2+1)!},
\]
which is the finite-moment analogue of the exact all-orders recovery in the infinite-dimensional moment-matching dilation. For the time-dependent case $H(t), K(t)$, this procedure can be applied with Dyson series expansions \cite{li2025linear}.

\section{The Derivation of the operator bound for HS/Magnus slices}\label{sm:Kmax_bound}

Our dilation primitive operates in a constant-success/window regime when the segment length $\tau$ satisfies
$\tau K_{\max}=\co(1)$, where $K_{\max}:=\max_{t\in[0,\tau]}\|K(t)\|$ (or $K_{\max}:=\max_n\|K_n\|$ in a discrete-time
implementation) \cite{li2025linear}. A sufficient condition is the explicit
bound $\tau\le (e\theta K_{\max})^{-1}$ \cite{li2025linear}.

Consider the Stratonovich SDE in the Feynman-Kac formula 
\[
d\ket{\phi_t}=\Big(-\hat H_1\,dt+\sum_{\gamma=1}^m\sqrt{2}\,\hat L_\gamma\circ dW_\gamma(t)\Big)\ket{\phi_t},
\]
and the associated one-step Magnus exponent $\Omega_n$ over $[t_n,t_{n+1}]$.

For weak order 1,
\[
\Omega^{(1)}_n=-\Delta t\,\hat H_1+\sum_{\gamma=1}^m\sqrt{2}\,\hat L_\gamma\,\Delta W_{\gamma,n},
\]
which is Hermitian when $\hat H_1,\hat L_\gamma$ are Hermitian and $\Delta W_{\gamma,n}\in\mathbb{R}$.
For weak order 2 with commuting diffusion fields, a correction term appears
\[
\Omega^{(2)}_n=\Omega^{(1)}_n-\sum_{\gamma=1}^m[\hat H_1,\hat L_\gamma]\,J_{0\gamma,n}.
\]
Since $[\hat H_1,\hat L_\gamma]$ is skew-Hermitian and $J_{0\gamma,n}\in\mathbb{R}$, the second-order term is purely
skew-Hermitian and therefore does not change the Hermitian part, $K_n$, of $\Omega_n$. Hence Magnus--1 and Magnus--2 have the
same contractive component:
\[
K_n:=-\Delta t\,\hat H_1+\sum_{\gamma=1}^m\sqrt{2}\,\hat L_\gamma\,\Delta W_{\gamma,n}.
\]

By the triangle inequality,
\[
\|K_n\|\le \Delta t\,\|\hat H_1\|+\sum_{\gamma=1}^m\sqrt{2}\,\|\hat L_\gamma\|\,|\Delta W_{\gamma,n}|.
\]
If we implement $\Delta W_{\gamma,n}$ via a bounded Hutchinson/Rademacher approximation \cite{kloeden1992numerical}
$\Delta W_{\gamma,n}=\sqrt{\Delta t}\,\xi_{\gamma,n}$ with $|\xi_{\gamma,n}|\leq \sqrt{3}$, then a pathwise uniform bound is found,
\[
\|K_n\|\le \Delta t\,\|\hat H_1\|+\sqrt{6\Delta t}\,\sum_{\gamma=1}^m\|\hat L_\gamma\|.
\]
This provides an explicit bound that determines the time interval to which the dilation method can be applied with constant success. For example, for weak order 1 scheme applied until $t=1$, we should have 
\[
\norm{\Omega}:=\norm{\Omega^{(1)}_n}/\dt \leq \|\hat H_1\|+\sqrt{6/\Delta t}\,\sum_{\gamma=1}^m\|\hat L_\gamma\|.
\]

\section{LCU coefficient normalization and constant success}\label{app:lcu_l1}

We recall the discrete spectral decomposition of the ancilla generator and the resulting LCU form.
Let $F_h=-i\sum_{j=0}^{M}\omega_j\ketbra{\phi_j}$ with $\omega_j\in\mathbb{R}$ and $\{\ket{\phi_j}\}$ orthonormal.
Then the dilated propagator factorizes in this eigenbasis and yields a system-only expansion
\begin{equation}\label{eq:lcu_expansion_app}
B(t)=\sum_{j=0}^{M} c_j\,U_j(t),\qquad
U_j(t):=\mathcal{T}\exp\!\left(-i\int_0^t (H(s)+\omega_jK(s))\,ds\right),
\end{equation}
with coefficients
\begin{equation}\label{eq:lcu_coeff_app}
c_j := \langle l_h|\phi_j\rangle\langle \phi_j|r_h\rangle.
\end{equation}
Here $(\langle l_h|,|r_h\rangle)$ are the windowed boundary vectors. For the window ${\rm win}=\{0,1,\dots,j^\ast\}$,
define the window weight
\begin{equation}\label{eq:pwin_app}
P_{\rm win}:=\sum_{j\in{\rm win}}|\langle j|r_h\rangle|^2
=\|\Pi_{\rm win}|r_h\rangle\|^2,
\qquad
\Pi_{\rm win}:=\sum_{j\in{\rm win}}\ketbra{j}.
\end{equation}
In the constant-success regime (segment length $t\|K\|=\co(1)$ with windowed readout), $P_{\rm win}=\Omega(1)$.

\begin{lemma}[$\ell_1$-bound for the LCU coefficients]\label{lem:lcu_l1}
With $c_j$ defined in \eqref{eq:lcu_coeff_app}, one has
\begin{equation}\label{eq:l1_bound}
\|c\|_1 := \sum_{j=0}^{M}|c_j|
\le \|l_h\|\,\|r_h\|.
\end{equation}
In particular, for the dilation-consistent normalization used in \eqref{eq:lcu_expansion_app} one may take $\|r_h\|=1$
and $\|l_h\|=1/\sqrt{P_{\rm win}}$, hence
\begin{equation}\label{eq:l1_pwin}
\|c\|_1 \le \frac{1}{\sqrt{P_{\rm win}}}=\co(1).
\end{equation}
\end{lemma}

\begin{proof}
By \eqref{eq:lcu_coeff_app} and Cauchy--Schwarz,
\[
\sum_{j=0}^{M}|c_j|
=\sum_{j=0}^{M}\big|\langle l_h|\phi_j\rangle\langle \phi_j|r_h\rangle\big|
\le
\Big(\sum_{j=0}^{M}|\langle l_h|\phi_j\rangle|^2\Big)^{1/2}
\Big(\sum_{j=0}^{M}|\langle \phi_j|r_h\rangle|^2\Big)^{1/2}.
\]
Since $\{\phi_j\}$ is an orthonormal basis, the two sums are $\|l_h\|^2$ and $\|r_h\|^2$, proving
\eqref{eq:l1_bound}. For the windowed/truncated choice, $\|r_h\|$ is normalized to $1$ and $l_h$ carries the window
renormalization $\|l_h\|=1/\|\Pi_{\rm win}r_h\|=1/\sqrt{P_{\rm win}}$, which gives \eqref{eq:l1_pwin}.
\end{proof}

\section{Extension to general two-body decompositions with both signs}
Our numerical tests have considered attractive two-body interactions.
In general, after normal ordering and a square decomposition of the interaction kernel, e.g., by diagonalization or Cholesky factorization, followed if needed by a Hermitian recombination of channels, the two-body part can be written as
\begin{equation}
\hat H_2=\frac12\sum_{\gamma=1}^{m}\lambda_\gamma\,\hat v_\gamma^2+\hat H_{\rm shift},
\end{equation}
where each \(\hat v_\gamma\) is one-body, i.e., quadratic in fermionic creation/annihilation operators, \(\lambda_\gamma\in\mathbb{R}\) may have either sign, and \(\hat H_{\rm shift}\) collects constant and one-body terms that can be absorbed into \(\hat H_1\) and the energy offset \cite{MottaZhangChan2019_AFQMCsym,ShiZhang2021_AFQMC_Review}.

For \(\lambda_\gamma<0\), the Hubbard--Stratonovich (HS) transformation introduces a real auxiliary field,
\begin{equation}
e^{-(\tau/2)\lambda_\gamma \hat v_\gamma^2}
=
\frac{1}{\sqrt{2\pi}}\int_{\mathbb{R}}dx_\gamma\,e^{-x_\gamma^2/2}\,
\exp\!\big(x_\gamma\sqrt{-\tau\lambda_\gamma}\,\hat v_\gamma\big),
\end{equation}
so each sample contributes a generally nonunitary one-body factor \(\exp(\eta_\gamma \hat v_\gamma)\) with real \(\eta_\gamma\). This is the case emphasized in the main text, where each slice is implemented by dilation or LCU.

For \(\lambda_\gamma>0\), the HS coupling is imaginary,
\begin{equation}
e^{-(\tau/2)\lambda_\gamma \hat v_\gamma^2}
=
\frac{1}{\sqrt{2\pi}}\int_{\mathbb{R}}dx_\gamma\,e^{-x_\gamma^2/2}\,
\exp\!\big(i\,x_\gamma\sqrt{\tau\lambda_\gamma}\,\hat v_\gamma\big).
\end{equation}
If \(\hat v_\gamma\) is Hermitian, the sampled noise factor is unitary. The full imaginary-time slice is still nonunitary because of the deterministic drift, and the Monte Carlo estimator acquires fluctuating complex phases, i.e., the usual phase problem \cite{ZhangCarlsonGubernatis1997_CPMC,ZhangKrakauer2003_Phaseless,ShiZhang2021_AFQMC_Review}.

Rather than HS transformation, this paper relies on the Feynman-Kac formula. 
Let
\begin{equation}
\hat H=\hat H_1+\frac12\sum_{\gamma=1}^m \lambda_\gamma \hat v_\gamma^2+\hat H_{\rm shift},
\end{equation}
and let \(\{W_\gamma(\tau)\}_{\gamma=1}^m\) be independent standard real Wiener processes. Define
\begin{equation}
\sigma_\gamma :=
\begin{cases}
\sqrt{-\lambda_\gamma}, & \lambda_\gamma<0,\\[2mm]
i\sqrt{\lambda_\gamma}, & \lambda_\gamma>0,
\end{cases}
\qquad\text{so that}\qquad
\sigma_\gamma^2=-\lambda_\gamma .
\end{equation}
Then the propagator-valued process \(\hat B(\tau)\) may be defined in Stratonovich form by
\begin{equation}
d\hat B(\tau)
=
-\big(\hat H_1+\hat H_{\rm shift}\big)\hat B(\tau)\,d\tau
+\sum_{\gamma=1}^m \sigma_\gamma\,\hat v_\gamma\,\hat B(\tau)\circ dW_\gamma(\tau),
\qquad \hat B(0)=\mathbb{I}.
\end{equation}

Pathwise, \(\hat B(\tau)\) is the time-ordered stochastic exponential
\begin{equation}
\hat B(\tau)
=
\mathcal{T}\exp\!\Big(
-\int_0^\tau (\hat H_1+\hat H_{\rm shift})\,ds
+\sum_{\gamma=1}^m \int_0^\tau \sigma_\gamma \hat v_\gamma \circ dW_\gamma(s)
\Big),
\end{equation}
and its expectation satisfies
\begin{equation}
\mathbb{E}\big[\hat B(\tau)\big]=e^{-\tau \hat H}.
\end{equation}
Thus the same stochastic representation covers both attractive and repulsive channels in one framework. When \(\lambda_\gamma<0\), the noise factors are generally nonunitary. When \(\lambda_\gamma>0\) and \(\hat v_\gamma\) is Hermitian, the noise factors are unitary but the estimator becomes oscillatory because of complex phases. This is precisely the setting in which constrained-path or phaseless controls are introduced in classical AFQMC \cite{ZhangCarlsonGubernatis1997_CPMC,ZhangKrakauer2003_Phaseless,ShiZhang2021_AFQMC_Review}.

\section{Extensions to higher-order stochastic Magnus terms}\label{app:magnus34}
We have mainly considered the first- and second-order Magnus expansion in the paper. 
This expansion can be continued to include higher-order stochastic integrals 
\cite{kamm2021stochastic,kloeden1992numerical}, and thus leads to more accurate stochastic paths.

Let \(U(t_{n+1},t_n)=\exp(\Omega_n)\) denote the one-step propagator, with stochastic differential
generator written in Stratonovich form as
\[
A(t)\,dt
=
-\hat H_{1}\,dt
+\sqrt{2}\sum_{\gamma=1}^m \hat v_\gamma \circ dW_\gamma(t).
\]
Formally,
\[
U(t_{n+1},t_n)
=
\mathcal{T}\exp\!\Big(\int_{t_n}^{t_{n+1}} A(s)\Big),
\qquad
\Omega_n=\Omega_{n,1}+\Omega_{n,2}+\Omega_{n,3}+\Omega_{n,4}+\cdots .
\]
The first three Magnus terms have the standard commutator form
\begin{align}
\Omega_{n,1} &=
\int_{t_n}^{t_{n+1}}A(s_1), \label{eq:mag1_app}\\
\Omega_{n,2} &=
\frac12\int_{t_n}^{t_{n+1}}\!\!\int_{t_n}^{s_1}[A(s_1),A(s_2)], \label{eq:mag2_app}\\
\Omega_{n,3} &=
\frac16\int_{t_n}^{t_{n+1}}\!\!\int_{t_n}^{s_1}\!\!\int_{t_n}^{s_2}
\Big([A(s_1),[A(s_2),A(s_3)]]
+[A(s_3),[A(s_2),A(s_1)]]\Big), \label{eq:mag3_app}
\end{align}
and \(\Omega_{n,4}\) is the corresponding fourth-order nested-commutator term
\cite{kamm2021stochastic}. In the Stratonovich setting these expressions have the same formal
structure as in the deterministic Magnus expansion, except that the coefficients are iterated
Stratonovich stochastic integrals.

For the density-coupled HS decomposition used in the main text, the diffusion operators commute,
\[
[\hat v_\gamma,\hat v_{\gamma'}]=0\qquad\forall \gamma,\gamma'.
\]
Hence all pure noise-noise commutator terms vanish. At second order this removes L\'evy-area
contributions and yields the order-2 generator in the paper. At higher orders, the only Lie
brackets that can appear are nested commutators involving \(\hat H_1\) and the
\(\hat v_\gamma\), for example
\[
[\hat H_{1},\hat v_\gamma],\qquad
[\hat H_{1},[\hat H_{1},\hat v_\gamma]],\qquad
[\hat v_{\gamma'},[\hat H_{1},\hat v_\gamma]],\ \ldots
\]
All of these remain one-body because commutators of one-body fermionic operators are one-body.
Therefore the third- and fourth-order Magnus truncations,
\[
\Omega_n^{(3)}:=\Omega_{n,1}+\Omega_{n,2}+\Omega_{n,3},
\qquad
\Omega_n^{(4)}:=\Omega_n^{(3)}+\Omega_{n,4},
\]
still generate fermionic Gaussian maps \(\exp(\Omega_n^{(3)})\) and
\(\exp(\Omega_n^{(4)})\) that preserve Slater-determinant structure.

In the commuting-noise setting, the additional random variables entering higher-order Magnus
truncations can be expressed in terms of a small number of time-integrated Brownian functionals
per channel, in addition to the Wiener increment
\[
\Delta W_{\gamma,n}=W_\gamma(t_{n+1})-W_\gamma(t_n)\sim \mathcal{N}(0,\Delta t).
\]
For example, on a single step \([0,\Delta t]\), they are given by, 
\begin{align}
J_{0\gamma} &:= \int_0^{\Delta t} W_\gamma(s)\,ds
= \int_0^{\Delta t}(\Delta t-s)\,dW_\gamma(s), \label{eq:J0_def_app}\\
J_{00\gamma} &:= \int_0^{\Delta t}\!\!\int_0^{s_1} W_\gamma(s_2)\,ds_2\,ds_1
= \int_0^{\Delta t}\frac{(\Delta t-s)^2}{2}\,dW_\gamma(s), \label{eq:J00_def_app}\\
J_{000\gamma} &:= \int_0^{\Delta t}\!\!\int_0^{s_1}\!\!\int_0^{s_2} W_\gamma(s_3)\,ds_3\,ds_2\,ds_1
= \int_0^{\Delta t}\frac{(\Delta t-s)^3}{6}\,dW_\gamma(s). \label{eq:J000_def_app}
\end{align}
These variables are linear functionals of the Brownian path and therefore are jointly Gaussian
with \(\Delta W_{\gamma,n}\). Their covariance matrix is obtained directly from the It\^o-integral
representations above via It\^o isometry \cite{kloeden1992numerical}. Thus higher-order Magnus
sampling requires only a modest Gaussian preprocessing layer, while the quantum step remains the
implementation of a one-body exponential.

When the diffusion operators do not commute, the second-order Magnus correction additionally involves L\'evy-area terms. However, these enter through commutators of one-body operators and therefore still produce an effective one-body generator, so a weak second-order method remains implementable within the same fermionic-Gaussian framework.

\section{Proof of the main Theorem }
 \label{app:proof_FT_main}

\begin{theorem}
Let \(\hat H\) be Hermitian with eigenvalues \(E_0<E_1\le\cdots\), gap
\(\Delta:=E_1-E_0>0\), and eigenvectors $\ket{E_0},\ket{E_1}, \cdots $. Let \(\ket{\psi_T}\) be a normalized trial wave function with ground overlap
\(\gamma:=|\braket{E_0}{\psi_T}|^2>0\), and define
\[
\ket{\psi(\tau)}
:=
\frac{e^{-\tau(\hat H-E_T)}\ket{\psi_T}}
{\|e^{-\tau(\hat H-E_T)}\ket{\psi_T}\|},
\qquad E_T\in\mathbb{R}.
\]
We choose a path-independent Magnus bound as in \cref{sm:Kmax_bound}
\[
\|\Omega_n(\omega)\|\le \norm{\Omega} \Delta t,
\]
and  a segment length \(\tau=n_T\Delta t\) with
\begin{equation}\label{tau-dt}
\tau=\Theta(1/\norm{\Omega}),
\qquad
\Delta t=\Theta(\sqrt{\epsilon_{\rm prep}}).
\end{equation}
Assume further that one segment admits a coherent fault-tolerant implementation of the normalized update map with uniform state-preparation error
\[
\sup_{\rho}
\big\|
\widetilde{\Phi}_\tau(\rho)-\Phi_\tau(\rho)
\big\|_1
\le c_{\rm prep}\,\epsilon_{\rm prep},
\]
where
\[
\Phi_\tau(\rho)
=
\frac{K_\tau \rho K_\tau}{\Tr(K_\tau \rho K_\tau)},
\qquad
K_\tau=e^{-\tau(\hat H-E_T)}.
\]

Then there is a fault-tolerant algorithm that, after
\[
S=\Theta\!\Big(\kappa \log \frac{1}{\gamma \epsilon}\Big),
\qquad
\kappa:=\frac{\norm{\Omega}}{\Delta},
\]
segments, prepares a normalized state \(\ket{\psi_S}\) satisfying
\[
\langle \psi_S|\hat H|\psi_S\rangle-E_0\le \epsilon,
\]
provided
\[
\epsilon_{\rm prep}
=
\mathcal O\!\left(
\frac{(1-e^{-2\tau\Delta})\,\epsilon}{\|\hat H\|}
\right).
\]

Moreover, given a block-encoding of \(\hat H/\alpha\) (\(\alpha\ge \|\hat H\|\)), the final energy
\(\langle \psi_S|\hat H|\psi_S\rangle\) can be estimated to additive error \(\epsilon\) with failure probability at most
\(\delta\) using
\[
\widetilde{\mathcal O}\!\big(\alpha\,\epsilon^{-1}\log(1/\delta)\big)
\]
coherent uses of the final state-preparation routine together with the block-encoding
\cite{Rall2021}.
\end{theorem}

Similar to standard analysis of numerical methods for differential equations,
we organize the proof in three parts: (i) one-segment local error, (ii)  contraction under imaginary-time projection, and (iii) the recursion that yields the global error bound. 

\paragraph{Local error.} Write the normalized, ideal, one-segment imaginary-time map as \[ \Phi_\tau(\rho) := \frac{K_\tau \rho K_\tau}{\Tr(K_\tau \rho K_\tau)}, \qquad K_\tau:=e^{-\tau(\hat H-E_T)}. \] 
We used a Kraus form for the density matrix that is commonly used in expressing imaginary time evolution. 
Since \(E_T\) is a scalar, the normalized map \(\Phi_\tau\) is independent of the choice of shift. Thus the ideal state after \(S\) segments is \[ \rho_S^\star=\Phi_\tau^S(\rho_0), \qquad \rho_0=\ketbra{\psi_T}. \]

Meanwhile,  one segment of our algorithm admits a coherent fault-tolerant approximation of the channel, denoted by \(\widetilde\Phi_\tau\). We choose the step size $\dt$, such that,  
\begin{equation} \sup_{\rho}\big\|\widetilde\Phi_\tau(\rho)-\Phi_\tau(\rho)\big\|_1 \le c_{\rm prep}\,\epsilon_{\rm prep}. \label{eq:uniform_seg_error_app} \end{equation} 
This bound is obtained from two contributions. First, the second-order Magnus truncation has local slice error \(\mathcal O(\Delta t^3)\), hence over one segment \[ n_T\,\mathcal O(\Delta t^3)=\mathcal O(\tau\Delta t^2)=\mathcal O(\epsilon_{\rm prep}), \] using \(\tau=\Theta(1/\norm{\Omega})\) and \(\Delta t=\Theta(\sqrt{\epsilon_{\rm prep}})\) \eqref{tau-dt}. Second, the dilation/LCU approximation is chosen so that the cumulative slice implementation error over one segment is also \(\mathcal O(\epsilon_{\rm prep})\), by using the light cone condition and increasing $n_A$. In the constant-success regime \(\tau=\Theta(1/\norm{\Omega})\), oblivious amplitude amplification yields a coherent normalized segment routine, so \eqref{eq:uniform_seg_error_app} holds uniformly in the input state.

\paragraph{Energy contraction.} This is the same gap-overlap contraction mechanism that underlies fault-tolerant imaginary-time and spectral-filtering algorithms for ground-state preparation; see, e.g., Refs.~\cite{LinTong2020nearoptimalground,DongLinTong2022QETU}. Let \[ P_0=\ketbra{E_0}{E_0}, \qquad g(\rho):=\Tr(P_0\rho), \qquad r(\rho):=\frac{1-g(\rho)}{g(\rho)}. \] 
For the initial trial state \(\rho_0=\ketbra{\psi_T}\), we have 
\[ g(\rho_0)=\gamma, \qquad r(\rho_0)=\frac{1-\gamma}{\gamma}. \]

More generally, if \(\rho=\sum_{j,k}\rho_{jk}\ketbra{E_j}{E_k}\), then \[ g(\Phi_\tau(\rho)) = \frac{\rho_{00}}{\rho_{00}+\sum_{k\ge 1}\rho_{kk}e^{-2\tau(E_k-E_0)}}, \] and therefore \[ r(\Phi_\tau(\rho)) = \frac{\sum_{k\ge 1}\rho_{kk}e^{-2\tau(E_k-E_0)}}{\rho_{00}} \le e^{-2\tau\Delta}\,r(\rho). \]

To proceed, we define \[ q:=e^{-2\tau\Delta}\in(0,1), \] 
we obtain after \(S\) ideal segments
\begin{equation} 
r(\rho_S^\star)\le q^S\,\frac{1-\gamma}{\gamma}. \label{eq:ideal_ratio_contract_app} \end{equation} 

To convert this into an energy bound, set 
\[ C_H:=\|\hat H-E_0 I\|\le 2\|\hat H\|. \] 
For any normalized \(\rho\) with  \(\rho=\sum_{j,k}\rho_{jk}\ketbra{E_j}{E_k}\),  the following holds
\[ \Tr[(\hat H-E_0 I)\rho] = \sum_{k\ge 1}\rho_{kk}(E_k-E_0) \le C_H\sum_{k\ge 1}\rho_{kk} = C_H\,(1-g(\rho)). \]
Hence we have 
\[ \Tr(\hat H\rho)-E_0 \le C_H\,\frac{r(\rho)}{1+r(\rho)} \le C_H\,r(\rho). \] Combining this with \eqref{eq:ideal_ratio_contract_app} yields the ideal contraction estimate \begin{equation} \Tr(\hat H\rho_S^\star)-E_0 \le C_H\,q^S\frac{1-\gamma}{\gamma}. \label{eq:ideal_energy_contract_app} \end{equation} 

\paragraph{Global error.} Again let \(\rho_{s+1}:=\widetilde\Phi_\tau(\rho_s)\) be the actually prepared state after \(s+1\) segments. Since the ideal map increases the ground weight and the segment error bound \eqref{eq:uniform_seg_error_app} is uniform, all states \(\rho_s\) remain in a compact set with \(g(\rho_s)\) bounded below by a constant depending only on \(\gamma\), provided \(\epsilon_{\rm prep}\) is sufficiently small. Thus, the map \[ g\mapsto r=\frac{1-g}{g} \] can be made Lipschitz. Therefore there exists a constant \(C_\gamma>0\), depending only on the initial overlap \(\gamma\), such that \[ r(\widetilde\Phi_\tau(\rho)) \le r(\Phi_\tau(\rho))+C_\gamma\,\epsilon_{\rm prep} \] for every normalized input \(\rho\). 

Comparing to the ideal contraction from the previous step, we obtain the recursion \[ r(\rho_{s+1}) \le q\,r(\rho_s)+C_\gamma\,\epsilon_{\rm prep}. \] 
Iterating gives 
\begin{equation} r(\rho_S) \le q^S\,\frac{1-\gamma}{\gamma} + \frac{C_\gamma}{1-q}\,\epsilon_{\rm prep}. \label{eq:r_recursion_final_app}
\end{equation}
Applying again the energy bound \(\Tr(\hat H\rho)-E_0\le C_H r(\rho)\), we obtain \begin{equation} 
\bra{\psi_S}\hat H\ket{\psi_S}-E_0 \le C_H\,q^S\frac{1-\gamma}{\gamma} + \frac{C_H C_\gamma}{1-q}\,\epsilon_{\rm prep}. 
\label{eq:energy_final_bound_app} 
\end{equation} 

We now choose parameters so that each term in \eqref{eq:energy_final_bound_app} is at most \(\epsilon/2\). For the contraction term, it suffices to take \[ S \ge \frac{1}{2\tau\Delta} \log\!\Big(\frac{2C_H(1-\gamma)}{\gamma\,\epsilon}\Big), \] which implies \[ S=\Theta\!\Big(\frac{\norm{\Omega}}{\Delta}\log\frac{1}{\gamma\epsilon}\Big) = \Theta(\kappa\log(1/(\gamma\epsilon))) \] under \(\tau=\Theta(1/\norm{\Omega})\).
The logarithmic dependence on \(\gamma^{-1}\) here pertains only to the number of imaginary-time segments needed to contract the excited-state weight. 

For the preparation error term, it suffices to require \[ \epsilon_{\rm prep} \le \frac{(1-q)\,\epsilon}{2C_H C_\gamma}, \qquad q=e^{-2\tau\Delta}, \] which is equivalent to \[ \epsilon_{\rm prep} = \mathcal O\!\left( \frac{(1-e^{-2\tau\Delta})\,\epsilon}{\|\hat H\|} \right). \] This proves the state-preparation part of the theorem. Finally, once the segment routine is available coherently, the full preparation of \(\ket{\psi_S}\) is obtained by composing the \(S\) segment unitaries and their ancilla restorations. Given in addition a block-encoding of \(\hat H/\alpha\), standard coherent expectation-estimation primitives imply that \(\bra{\psi_S}\hat H\ket{\psi_S}\) can be estimated to additive error \(\epsilon\) with failure probability at most \(\delta\) using \[ \widetilde{\mathcal O}\!\left(\alpha\,\epsilon^{-1}\log\frac{1}{\delta}\right) \] controlled uses of the final state-preparation routine together with the block-encoding \cite{Rall2021}. This proves the theorem.
\end{document}